\documentclass{lmcs} 
\pdfoutput=1

\usepackage{lastpage}
\lmcsdoi{16}{1}{12}
\lmcsheading{}{\pageref{LastPage}}{}{}%
{Jan.~02,~2019}{Feb.~11,~2020}{}

\theoremstyle{defC}
\newtheorem{exaC}[thm]{Example}


\usepackage{tikz}
\usepackage{amssymb,amsmath,amsthm}

\input{edit.tex}
\usepackage{macros}
\usepackage{old}

\usepackage{hyperref}

\theoremstyle{plain}\newtheorem{theorem}[thm]{Theorem} 
\theoremstyle{definition}\newtheorem{example}[thm]{Example} 
\theoremstyle{plain}\newtheorem{lemma}[thm]{Lemma} 

\begin{document}

\title[Undecidability of a weak version of \msoU]{Undecidability of a weak version of \texorpdfstring{\msoU}{MSO+U}}

\author[M. Bojańczyk]{Mikołaj Bojańczyk}	
\address{University of Warsaw, Poland}	
\email{bojan@mimuw.edu.pl}  

\author[L.~Daviaud]{Laure Daviaud}	
\address{City, University of London, United Kingdom}	
\email{laure.daviaud@city.ac.uk}  

\author[B.~Guillon]{Bruno Guillon}	
\address{LIMOS, Université Clermont Auvergne, Aubière, France}	
\email{bruno.guillon@uca.fr}  

\author[V.~Penelle]{Vincent Penelle}	
\address{University of Bordeaux, CNRS,  Bordeaux INP, LaBRI, UMR 5800, F-33400, Talence, France}	
\email{vincent.penelle@labri.fr}  

\author[A. V. Sreejith]{A. V. Sreejith}	
\address{IIT Goa, India}	
\email{sreejithav@iitgoa.ac.in}  





\begin{abstract}
	We prove the undecidability of \mso on $\omega$-words extended with the second-order predicate \weakU[X]
	which says that the distance between consecutive positions in a set $\Asetvar\subseteq\Nat$ is unbounded. 
	This is achieved by showing that adding \weakU to \mso gives a logic with the same expressive power as \msoU,
	a logic on $\omega$-words with undecidable satisfiability.
	As a corollary, we prove that \mso on $\omega$-words becomes undecidable
	if allowing to quantify over sets of positions that are ultimately periodic,
	\ie, sets \Asetvar such that for some positive integer $p$,
	ultimately either both or none of positions $x$ and $x+p$ belong to \Asetvar.
\end{abstract}

\maketitle

\section{Introduction}

This paper is about monadic second-order logic (\mso) on $\omega$-words.
B\"uchi's famous theorem says that
given an \mso sentence describing a set of $\omega$-words over some alphabet,
one can decide if the sentence is true in at least one $\omega$-word
\cite{buchi62}.
B\"uchi's theorem
along with its proof using automata techniques
have been the inspiration for a large number of decidability results for variants of \mso,
including Rabin's theorem on the decidability of \mso on infinite trees \cite{Rabin69}.

By now it is quite well understood that \mso is a maximal decidable logic over finite words. One formalisation of this can be found in \cite{LZ} (Theorem 9), which shows that the class of regular languages is the maximal class which: (a) has decidable satisfiability; and (b) is closed under Boolean operations and images under rational relations. The general idea behind the result in \cite{LZ} is that a non-regular language has a Myhill-Nerode equivalence relation of infinite index, and this together with closure under Boolean operations and images under rational relations can be used to generate counters, Turing machines, and then arbitrary languages in the arithmetic hierarchy. However, for languages of infinite words, the situation is different as logic over omega-words can talk about asymptotic properties, and getting a counter from asymptotic properties can be much harder.
One of the themes developed in the wake of B\"uchi's result is the question:
what can be added to \mso on $\omega$-words so that the logic remains decidable?
One direction, studied already in the sixties,
has been to extend the logic with predicates such as
\emph{``position $x$ is a square number''} or \emph{``position $x$ is a prime number''}.
See~\cite{Bes02,KLM15} and the references therein
for a discussion on this line of research.
Another direction,
which is the one taken in this paper,
is to study quantifiers which bind set variables.
These new quantifiers may for instance
talk about the number of sets satisfying a formula,
hence being midway between the universal and the existential quantifiers,
like the quantifier \emph{``there exists uncountably many sets''}~\cite{BKR11}
(which \aposteriori does not extend the expressivity of \mso
on finitely branching trees).
Another way is to consider quantifiers that talk about the asymptotic behaviors of infinite sets.
This was already the direction followed in \cite{bcc14},
where the authors define \emph{\amso}, a logic which talks about the asymptotic behaviors of sequences of integers
in a topological flavour.
Another example is \emph{the quantifier \quantifierU}
which was introduced in~\cite{Boj04}
and that says that some formula \Aform[X] holds for arbitrarily large finite set~$X$:
\begin{equation*}
	\quantifierU X\Aform[X]:
	\text{``for all $k \in \mathbb{N}$, \Aform[X] is true for some finite set $X$ of size at least $k$''}
	\text.
\end{equation*}
However, in~\cite{BPT16} it was shown that
\msoU,
namely \mso extended with the quantifier \quantifierU,
has undecidable satisfiability;
see~\cite{Boj15} for a discussion on this logic.

In this work, we study a, \apriori, weaker version of \msoU, the logic \msoW,
namely \mso extended with \emph{the second-order predicate \weakU} defined by:
\begin{equation*}
	\text{\weakU[X]}:
	\text{``for all $k \in \mathbb{N}$, there exist two consecutive positions of $X$ at distance at least $k$''}
	\text.
\end{equation*}

\begin{example}
\ 

\begin{itemize}
\item The set $X = \{10 n\mid n\in \mathbb{N}\}$ does not satisfy \weakU, as for every $n$, $10(n+1) - 10n = 10$, so there are no two consecutive positions at distance at least $11$.
\item The set $X = \{2^n \mid n\in \mathbb{N}\}$ does satisfy \weakU, as for every $k$, $2^{k+1} - 2^{k} > k$, and $2^{k+1}$ and $2^k$ are consecutive in $X$.
\item The set $X = \{10m + n \mid \text{n is the }m^{\text{th}}\text{ digit of }\pi\}$ does not satisfy \weakU, the difference between two consecutive elements being at most $19$, as there is always an element between $10m$ and $10m+9$.
\end{itemize}
\end{example}

\begin{exaC}[\cite{BC06}]
	Consider the language of $\omega$-words of the form $a^{n_1}ba^{n_2}b\cdots$
	such that $\set{n_1,n_2,\ldots}$ is unbounded.
	This language can be defined in \msoU by a formula
	saying that there are factors of consecutive $a$'s of arbitrarily large size.
	It can also be defined in \msoW
	saying that the set of the positions labeled by $b$ satisfies the predicate \weakU.
\end{exaC}

It is easy to see that the predicate \weakU can be defined in the logic \msoU:
a set $X$ satisfies \weakU[X]
\iof there exist intervals (finite connected sets of positions) of arbitrarily large size
which are disjoint with $X$.
Therefore, the logic \msoW can be seen as a fragment of \msoU.
Is this fragment proper?
The main contribution of this paper
is showing that actually the two logics are the same:
\begin{theorem}\label{thm:main}
	The logics  \msoU and \msoW define the same languages of $\omega$-words, and translations both ways are effective.
\end{theorem}

We believe
that \msoW can be reduced to many extensions of \mso, in a simpler way
than reducing \msoU.

As an example, we consider a quantifier which talks about ultimately periodic sets.
A set of positions $X\subseteq\Nat$ is called \emph{ultimately periodic}
if there is some period $\Aperiod\in\Nat$ such that for sufficiently large positions $x\in\Nat$,
either both or none of $x$ and $x+\Aperiod$ belong to $X$.
For example, the set $\{10n \mid n\in\mathbb{N}\}$ is periodic (with period 10), but the set $\{10m + n \mid \text{n is the }m^{\text{th}}\text{ digit of }\pi\}$ is not periodic.
We consider the logic \msoP,
\ie, \mso augmented with \emph{the quantifier \quantifierP} that ranges over ultimately periodic sets:
\begin{equation*}
	\quantifierP[X]\Aform[X]:
	\text{``the formula \Aform[X] is true for all ultimately periodic sets $X$''}.
\end{equation*}
Though quantifier \quantifierP extends the expressivity of \mso,
it is \apriori not clear whether it has decidable or undecidable satisfiability.
However, using Theorem~\ref{thm:main}, we obtain the following result, which answers an open question raised in~\cite{BC06}.
\begin{theorem}
	\label{thm:periodic}
	Satisfiability over $\omega$-words is undecidable for \msoP.
\end{theorem}

Another example is the logic \mso+probability as defined in \cite{MM}. This line of thought is not developed in the present article, but we can make the following reasoning.  A set X satisfies $\weakU(X)$ if and only if there exists an infinite set $Y$ of blocks in $X$ with the following property: there is nonzero probability of choosing a subset $Z \subseteq \mathbb{N}$ such that in every block from $Y$, there is at least one element of $Z$.  Therefore, \mso+probability is more expressive than $\msoW$.


\subsection*{Outline of the paper.}
The rest of the paper is mainly devoted to the proof of Theorem~\ref{thm:main}.
In Section~\ref{sec:level0} we introduce an intermediate logic \msoS.
We first prove that \msoU and \msoS are effectively equivalent
(Section~\ref{ssec:from-qU-to-pU2}, Lemma~\ref{lem:u-predicate}).
Then, we prove that \msoS and \msoW are also effectively equivalent
(Section~\ref{ssec:from-pU2-to-pU1}, Lemma~\ref{lem:u-predicates}),
assuming a certain property, namely Lemma~\ref{lem:msowdef:dimtoinf}.
The proof of this latter lemma is the subject of Section~\ref{sec:unbounded dimensions}.
Finally, in Section~\ref{sec:projection},
we discuss the expressive power of \msoP
with respect to \msoU. We prove Theorem~\ref{thm:periodic} and we show that the property ``ultimately periodic'' can be expressed in \msoU
if allowing a certain encoding.


\section{How \texorpdfstring{MSO+\ensuremath{\texttt U_1}}{MSO+U1} talks about vector sequences}
\label{sec:level0}
In this section,
we introduce a new predicate on pairs of sets of positions \strongU[R,I]
and consider the logic \mso extended with this predicate: \msoS.
We first prove in Section~\ref{ssec:from-qU-to-pU2},
that \msoU has the same expressive power as \msoS.
Then, we show in Section~\ref{ssec:from-pU2-to-pU1},
that \msoS is itself as expressive as \msoW,
assuming a certain property, namely Lemma~\ref{lem:msowdef:dimtoinf}.
The proof of this lemma is given in Section~\ref{sec:unbounded dimensions}.

\subsection{From quantifier \ensuremath{\texttt U} to predicate \texorpdfstring{\ensuremath{\texttt U_2}}{U2}}
\label{ssec:from-qU-to-pU2}
We say that a sequence of natural numbers is \emph{unbounded}
if arbitrarily large numbers occur in it.
For two disjoint sets of positions $R, I \subseteq \Nat$
with $R$ infinite, we define a sequence of numbers as in Figure~\ref{fig:first-encoding}:
the $i$-th element of this sequence is the number of elements from $I$ between the $i$-th
and the $(i+1)$-th elements from $R$.
In particular, elements of $I$ in positions smaller than all the positions from $R$ are not relevant.
\begin{figure}[h]
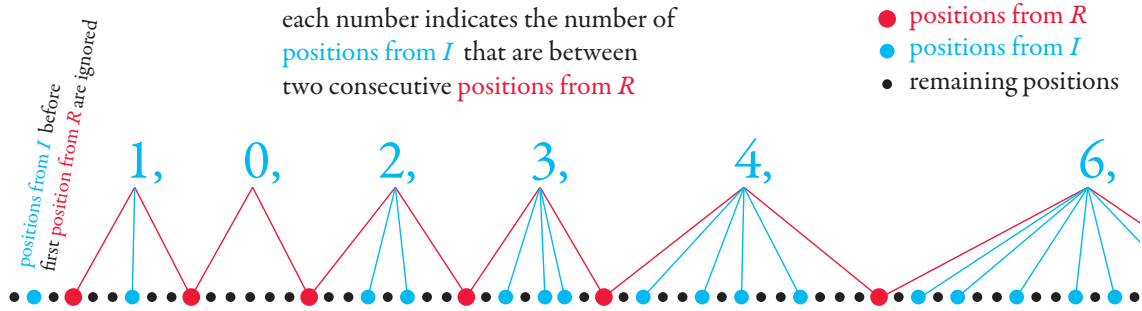

	\mypic{3}
	\caption{%
		Two sets of positions $R,I \subseteq \Nat$ and the sequence in $\Nat^\omega$ that they define.
		The sequence is only defined when $R$ and $I$ are disjoint, and $R$ is infinite.%
	}
	\label{fig:first-encoding}
\end{figure}

We now define the binary predicate \strongU:
\begin{equation*}
	\strongU[R,I]:
	\text{
		``the sequence of numbers encoded by $\Asetreset,\Asetinc$
		is defined and unbounded''%
	}.
\end{equation*}
The difference between the predicate \weakU,
which has one free set variable,
and the predicate \strongU,
which has two free set variables,
is that the latter is able to ignore some positions. It is illustrated
by the following example,
which can be easily defined in \msoS.
\begin{example}
	Let $L$ be the union of the languages of $\omega$-words
	of the form $(ac^*)^{n_1}b(ac^*)^{n_2}b\cdots$
	where $\set{n_1,n_2,\ldots}$ is unbounded.
	It can be defined in \msoS in the following way:
	an $\omega$-word belongs to $L$ \iof
	it belongs to $\big({(ac^*)}^*b\big)^\omega$
	and the sequence of numbers encoded by $R=\set{\text{set of positions labeled by $b$}}$
	and $I=\set{\text{set of positions labeled by $a$}}$ is unbounded.
	The language also admits a simple definition in \msoU:
	there exist arbitrarily large finite sets of positions labeled by $a$
	such that no two positions from the set are separated by a $b$.
	It is however not straightforward to define it in the logic \msoW,
	since in every $\omega$-word,
	the distance between two successive $b$'s separated by at least one $a$
	can be increased by adding occurrences of $c$'s,
	without changing the membership of the word to the language.
\end{example}

The following lemma follows essentially from the same argument as in~\cite[Lemma 5.5]{BC06}.
\begin{lemma}\label{lem:u-predicate}
	The logics \msoU and \msoS
	define the same languages of $\omega$-words,
	and translations both ways are effective.
\end{lemma}
\begin{proof}
	The predicate \strongU is easily seen to be expressible in \msoU:
	a sequence of numbers encoded by $R$ and $I$ is unbounded if the
	subsets of elements of $I$ between two consecutive elements of $R$
	are of arbitrarily large size, which is expressible in \msoU.
	For the converse implication, we use the following observation.
	A formula $\quantifierU[X]\Aform[X]$ is true if and only if the following condition holds:
	\begin{description}
		\item[{\namedlabel[$\mathbf{(\star)}$]{it:starstar}{$(\star)$}}]
			There exist two sets $\Asetreset,\Asetinc\subseteq\Nat$
			which satisfy \strongU[R,I]
			such that for every two consecutive positions $\Avarreset,\Avarresetbis\in\Asetreset$,
			there exists some finite \Asetvar satisfying \Aform[X]
			which contains all positions of \Asetinc between \Avarreset and \Avarresetbis.
			Here is a picture of property \ref{it:starstar}:

			\noindent\mypic{1}
	\end{description}
	Condition \ref{it:starstar} is clearly expressible in \mso with the predicate \strongU.
	Hence the lemma will follow once we prove the equivalence:
	$\quantifierU[X]\Aform(X)$ \iof \ref{it:starstar}.
	The right-to-left implication is easy to see.
	For the converse implication, we do the following construction.
	We define a sequence of positions
	$0=x_0<x_1<\ldots$
	as follows by induction.
	Define $x_0$ to be the first position,
	\ie,~the number $0$.
	Suppose that $x_n$ has already been defined.
	By the assumption $\quantifierU[X]\Aform[X]$, there exists a set $X_{n+1}$ which satisfies $\Aform$
	and which contains at least $n$ positions after $x_{n}$.
	Define $x_{n+1}$ to be the last position of $X_{n+1}$.
	This process is illustrated in the following picture:

	\noindent\mypic{4}
	Define $R$ to be all the positions $x_0,x_1,\ldots$ in the sequence thus obtained,
	and define $I$ to be the set of positions $x$
	such that $x_n<x<x_{n+1}$ and $x \in X_{n+1}$ holds for some $n$.
	By construction, the sets $I$ and $R$ thus obtained will satisfy \strongU.
\end{proof}

\subsection{From predicate \texorpdfstring{\ensuremath{\texttt U_2}}{U2} to predicate \texorpdfstring{\ensuremath{\texttt U_1}}{U1}}
\label{ssec:from-pU2-to-pU1}
Lemma~\ref{lem:u-predicate} above
states that \msoU and \msoS have the same expressive power,
thus giving a first step towards the proof of Theorem~\ref{thm:main}.
The second step, which is the key point of our result,
is the following lemma,
which states that \msoS is as expressive as \msoW.
\begin{lemma}
	\label{lem:u-predicates}
	The logics \msoS and \msoW
	define the same languages of $\omega$-words,
	and translations both ways are effective.
\end{lemma}
Note that one direction is straightforward,
since \weakU[X] holds \iof \strongU[X,\mathbb N\setminus X] holds.
Hence, it remains to show that the predicate \strongU
can be defined by a formula of the logic \msoW with two free set variables.
\smallbreak

In our proof, we use terminology and techniques about sequences of vectors of natural numbers
that were used in the undecidability proof for \msoU~\cite{BPT16}.
A \emph{vector sequence} is defined to be an element of $\left(\Nat^*\right)^\omega$,
\ie, a sequence of possibly empty tuples of natural numbers.
For a vector sequence $\Avecseq\in\left(\Nat^*\right)^\omega$,
we define its \emph{dimension}, denoted by $\dimens\Avecseq\in\Nat^\omega$,
as being the number sequence of the dimensions of the vectors:
the $i$-th element in \dimens\Avecseq is \emph{the dimension of the $i$-th vector} in \Avecseq,
\ie,~the number of coordinates in the $i$-th tuple in \Avecseq.
We say that a vector sequence \Avecseq \emph{tends towards infinity}, denoted by \toinf\Avecseq,
if every natural number appears in finitely many vectors from \Avecseq.

Recall Figure~\ref{fig:first-encoding},
which showed how to encode a number sequence using two sets of positions $\Asetreset,\Asetinc\subseteq\Nat$.
We now show that the same two sets of positions can be used to define a vector sequence.
As it was the case in Figure~\ref{fig:first-encoding},
we assume that \Asetreset is infinite and disjoint from \Asetinc.
Under these assumptions, we write $\vecseqofsets\Asetreset\Asetinc\in\left(\Nat^*\right)^\omega$
for the vector sequence defined according to the description from Figure~\ref{fig:encoding},
\ie, the sequence of vectors whose coordinates are the lengths of intervals in \Asetinc
between two consecutive elements of \Asetreset.
Remark that $0$ is not encoded
and thus that $\vecseqofsets\Asetreset\Asetinc$ contains possibly empty vectors
with only positive coordinates.
As for Figure~\ref{fig:first-encoding},
the positions of $I$ smaller than all the positions of $R$ are not relevant.
\begin{figure}[h]
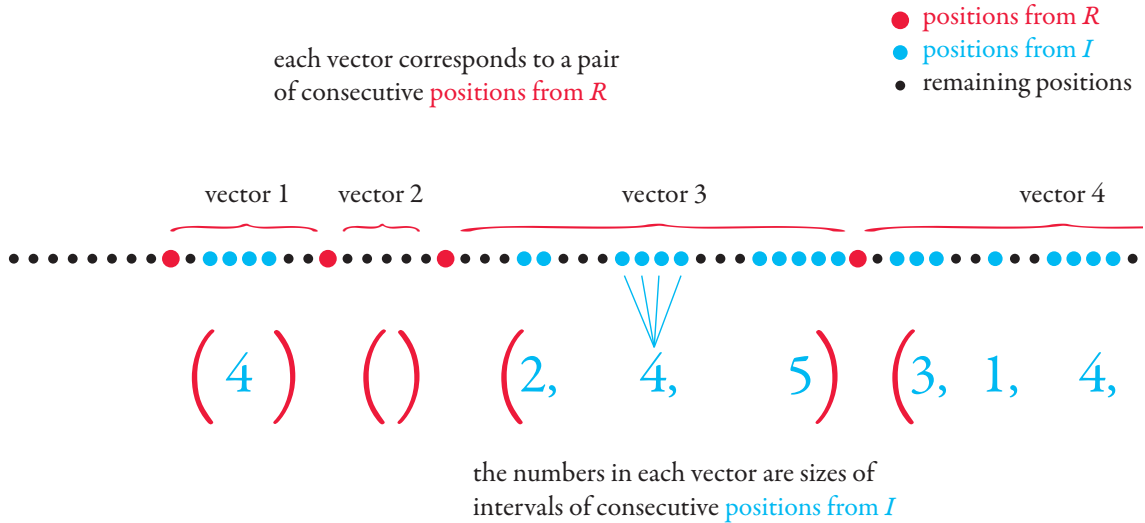

	\mypic{2}%
	\caption{%
		Two sets of positions $\Asetinc,\Asetreset\subseteq\Nat$
		and the vector sequence that they encode.
		This vector sequence is defined only when \Asetinc and \Asetreset are disjoint and \Asetreset is infinite.%
	}%
	\label{fig:encoding}%
\end{figure}
\smallbreak

The following lemma states
that the unboundedness of the dimensions of a vector sequence can be expressed in \msoW.
Its proof will be the subject of Section~\ref{sec:unbounded dimensions}.
\begin{lemma}
	\label{lem:msowdef:dimtoinf}
	There is a formula \Aform[R,I] in \msoW
	which is true \iof
	the vector sequence \vecseqofsets\Asetreset\Asetinc is defined
	and satisfies \dimens{\vecseqofsets\Asetreset\Asetinc} is unbounded.
\end{lemma}
\begin{proof}[Proof of Lemma~\ref{lem:u-predicates} assuming Lemma~\ref{lem:msowdef:dimtoinf}]
	As the predicate \weakU is easily definable in \msoS
	(\weakU[X] holds \iof \strongU[X,X\setminus\Nat] holds),
	it remains to prove the converse,
	\ie, to exhibit a formula \Aform[X,Y] of \msoW
	which holds \iof \strongU[X,Y] holds.

	Since \strongU[X,Y] holds if and only if \strongU[X,Y\setminus X] holds
	and because \strongU[X,Y] implies that $X$ is infinite,
	it is sufficient to write a formula \Aform[\Asetreset,\Asetinc] which holds
	\iof \strongU[\Asetreset,\Asetinc] holds,
	assuming \Asetreset and \Asetinc are disjoint
	and \Asetreset is infinite.
	We fix such \Asetreset and \Asetinc,
	which hence define a vector sequence \vecseqofsets\Asetreset\Asetinc
	according to the encoding of Figure~\ref{fig:encoding}.
	Furthermore,
	\strongU[R,I] holds
	\iof summing all coordinates of each vector from the sequence
	yields a number sequence which is unbounded.
	We show that we can construct a $\Asetincbis\subseteq\Asetinc$
	such that \strongU[R,I] holds
	\iof the dimension of \vecseqofsets\Asetreset\Asetincbis is unbounded.%

	For every non negative integer $x$,
	if $2x$ belongs to \Asetinc
	as well as
	either $2x+1$ or $2x-1$,
	then we remove $2x$ from \Asetinc.
	Let \Asetincbis denote the thus obtained set.
	As at least one position from \Asetinc over two is kept in \Asetincbis,
	we have that \strongU[\Asetreset,\Asetinc] holds
	\iof
	\strongU[\Asetreset,\Asetincbis] holds.
	Furthermore, \Asetincbis
	has the property that each of its point is isolated
	whence the sum of the coordinates of a vector from \vecseqofsets\Asetreset\Asetincbis
	is equal to its dimension.
	Therefore,
	\strongU[\Asetreset,\Asetincbis] holds \iof
	the dimension of the vector sequence \vecseqofsets\Asetreset\Asetincbis is unbounded.
	Moreover, \Asetincbis is definable in \mso
	in the sense that there is a \mso-formula with two free set variables \Asetinc and \Asetincbis
	which holds if and only if \Asetincbis is obtained from \Asetinc by the process given above.

	We conclude the proof of Lemma~\ref{lem:u-predicates}
	by applying Lemma~\ref{lem:msowdef:dimtoinf}.
%
\end{proof}


\section{Expressing unboundedness of dimensions}
\label{sec:unbounded dimensions}
This section is devoted to proving Lemma~\ref{lem:msowdef:dimtoinf}.
Our proof has two steps.
In Section~\ref{ssec:level1},
we show that in order to prove Lemma~\ref{lem:msowdef:dimtoinf},
it suffices to show the following lemma.
\begin{lemma}
	\label{lem:msowdef:ttoinf_and_dimttoinf}
	There is a formula \Aform[\Asetreset, \Asetinc] in \msoW
	which is true \iof
	the vector sequence \vecseqofsets\Asetreset\Asetinc is defined and satisfies
			\toinf{\vecseqofsets\Asetreset\Asetinc} 
			\text{and} 
			\dimens{\vecseqofsets\Asetreset\Asetinc}\text{is unbounded.}
\end{lemma}
Lemma~\ref{lem:msowdef:ttoinf_and_dimttoinf} itself will be proved in Section~\ref{ssec:level2}.

\subsection{From Lemma~\ref{lem:msowdef:ttoinf_and_dimttoinf} to Lemma~\ref{lem:msowdef:dimtoinf}}
\label{ssec:level1}
Recall that the vector sequence \vecseqofsets\Asetreset\Asetinc is defined
if and only if $R$ is infinite and disjoint from a set $I$,
which is clearly expressible in \mso.
So from now on, we will always consider vector sequences \vecseqofsets\Asetreset\Asetinc which are defined. 
\bigbreak

Given a vector sequence, a \emph{sub-sequence} is an infinite vector sequence obtained by dropping some of the vectors, like this:

\noindent\mypic{7}
An \emph{extraction} is obtained by removing some of the coordinates in some of the vectors,
but keeping at least one coordinate from each vector, like this:

\noindent\mypic{8}
In particular a vector sequence containing at least one empty vector has no extraction.
An extraction is called \emph{interval-closed}
if the coordinates which are kept in a given vector are consecutive, like this:

\noindent\mypic{9}
A \emph{$1$-extraction} is an extraction with only vectors of dimension $1$.
In particular, it is interval-closed.
A \emph{sub-extraction} is an extraction of a sub-sequence.
%
\medbreak

The following lemma gives two conditions equivalent to the unboundedness of the dimension of a vector sequence.
\begin{lemma}
	\label{lem:unbounded degree}
	\label{lem:unbounded-dim}
	Let \Avecseq be a vector sequence.
	Then \dimens\Avecseq is unbounded \iof
	one of the following conditions is true:
	\begin{enumerate}[ref={\textcolor{darkgray}{\textbf{(\arabic*)}}},label={\textcolor{darkgray}{\textbf{(\arabic*)}}},leftmargin=2em,nosep]
		\item
			\label{it:tinfty+unbounded-dim}
			There exists a sub-extraction \Avecseqbis of \Avecseq
			such that \toinf\Avecseqbis and \dimens\Avecseqbis is unbounded.
		\item
			\label{it:bounded+unbunded-dim}
			There exists an interval-closed sub-extraction \Avecseqbis of \Avecseq
			such that
			\dimens\Avecseqbis is unbounded
			and
			\Avecseqbis is bounded,
			\ie, there exists a bound $\Abound\in\Nat$
			such that all the coordinates occurring in the sequence are less than \Abound.
	\end{enumerate}
\end{lemma}
\begin{proof}
	The if implication is clear,
	since admitting a sub-extraction of unbounded dimension
	trivially implies the unboundedness of the dimension.
	
	We now focus on the only-if implication.
	Assume that \dimens\Avecseq is unbounded
	and condition~\ref{it:tinfty+unbounded-dim} is false.
	Then,
	for all sub-extractions \Avecseqbis of \Avecseq
	we have that
	\toinf\Avecseqbis implies \dimens\Avecseqbis is bounded.
	This implies that
	there exists a constant \Athreshold
	such that for all vectors \Avect in \Avecseq,
	the number of coordinates greater than \Athreshold
	is less than \Athreshold.
	Indeed, if no such \Athreshold exists,
	we can exhibit a sub-extraction \Avecseqbis of \Avecseq,
	which tends towards infinity and has unbounded dimension.
	Let us mark the coordinates smaller than \Athreshold in every vector of \Avecseq.
	For every integer \ANinteger, we can find a vector of \Avecseq with at least \ANinteger consecutive marked coordinates.
	Indeed suppose there is at most \ANinteger consecutive marked coordinates in every vector.
	Then, as there are at most \Athreshold unmarked coordinates, a vector cannot have dimension
	greater than $\ANinteger(\Athreshold+1) + \Athreshold$ (\ie, $\Athreshold + 1$ blocks of \ANinteger marked positions
	plus \Athreshold unmarked positions separating them).
	This contradicts that \dimens \Avecseq is unbounded.
	
	Thus by keeping for each vector one of the longest block of consecutive marked positions
	and removing vectors without any marked position,
	we can construct an interval-closed sub-extraction \Avecseqbis of \Avecseq which is 
	bounded by \Athreshold and has unbounded dimension.
\end{proof}

We now proceed to show how Lemma~\ref{lem:msowdef:ttoinf_and_dimttoinf} implies Lemma~\ref{lem:msowdef:dimtoinf}.
Recall that Lemma~\ref{lem:msowdef:dimtoinf} says that there is a formula \Aform[R,I] in \msoW
which is true if and only if the vector sequence \vecseqofsets\Asetreset\Asetinc is defined and satisfies
\dimens{\vecseqofsets\Asetreset\Asetinc} is unbounded.
\begin{proof}[Proof of Lemma~\ref{lem:msowdef:dimtoinf} assuming Lemma~\ref{lem:msowdef:ttoinf_and_dimttoinf}]
	Consider \Asetreset and \Asetinc two sets of positions such that \vecseqofsets\Asetreset\Asetinc is defined.
	We will use Lemma~\ref{lem:unbounded-dim} to express in \msoW that \dimens{\vecseqofsets\Asetreset\Asetinc} is unbounded.
	The first condition of the lemma is expressible in \msoW by Lemma~\ref{lem:msowdef:ttoinf_and_dimttoinf},
	as shown in the next section.
	However, in order to express the second condition,
	we need that the gaps between consecutive intervals of $I$ have bounded length.
	This is ensured by extending the intervals as explained in the following.
	For any two disjoint sets $R$ and $I$, let us construct $J$ as follows:
	$J$ contains $I$, and
	for every two consecutive positions $x<y$ in $I$ such that there is no element of $R$ in between,
	we add to $J$ all the positions strictly between $x$ and $y$,
	except $y-1$, as in Figure~\ref{fig:completion}.
	The set $J$ is expressible in \mso, in the sense that there is a \mso-formula with three free set variables $R$, $I$, $J$
	which holds if and only if $J$ is obtained from $I$ by this process.
	Moreover, it is clear that \dimens{\vecseqofsets\Asetreset\Asetinc} is unbounded if and only if 
	\dimens{\vecseqofsets{\Asetreset}{J}} is unbounded
	since the dimension in both sequences are pairwise equal
	(see Figure~\ref{fig:completion}).
	It suffices thus to prove Lemma~\ref{lem:msowdef:dimtoinf} for sets $R$ and $I$ which satisfy:
	\begin{description}
		\item[{\namedlabel[$\mathbf{(\star)}$]{it:star}{$(\star)$}}]
			between two consecutive elements of $I$,
			there is either an element of $R$ or at most one position not in $I$.
	\end{description}
	\begin{figure}[h]
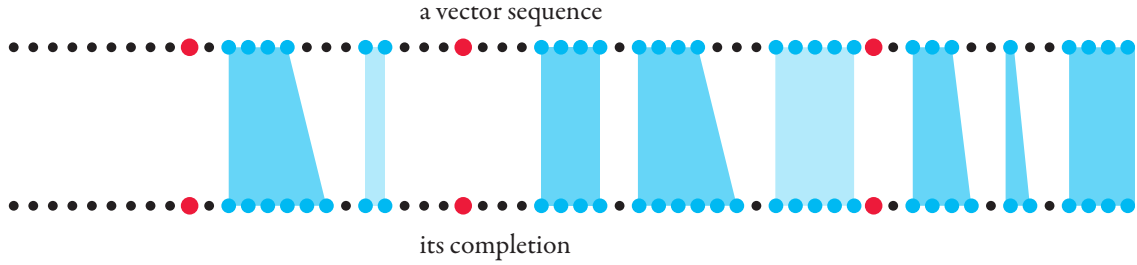

		\mypic{6}
		\caption{Completion of a vector sequence.}
		\label{fig:completion}
	\end{figure}
	\medbreak

	Let $R$ and $I$ be two sets satisfying \ref{it:star}.
	We want to express that \dimens{\vecseqofsets\Asetreset\Asetinc} is unbounded.
	It suffices to show that for each condition in Lemma~\ref{lem:unbounded degree},
	there exists a formula of \msoW which says that the condition
	is satisfied by \vecseqofsets\Asetreset\Asetinc.
	We treat the two cases separately.
	\begin{enumerate}[leftmargin=0em,itemindent=1.5em]
		\item We want to say that there exists a sub-extraction \Avecseqbis of \vecseqofsets\Asetreset\Asetinc
			such that \toinf\Avecseqbis and \dimens\Avecseqbis is unbounded.
			Sub-extraction can be simulated in \mso according to the picture in Figure~\ref{fig:sub-extraction},
			and the condition ``\toinf\Avecseqbis and \dimens\Avecseqbis is unbounded''
			can be checked using Lemma~\ref{lem:msowdef:ttoinf_and_dimttoinf}.
			\begin{figure}[h]
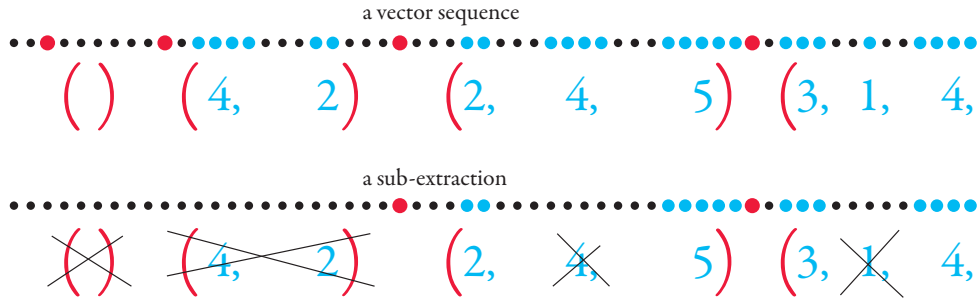

				\centering
				\mypic[.85]{5}
				\caption{Sub-extraction of a vector sequence.}
				\label{fig:sub-extraction}
			\end{figure}
		\item We want to say that
			there exists an interval-closed sub-extraction \Avecseqbis of \vecseqofsets\Asetreset\Asetinc
			such that \Avecseqbis is bounded and \dimens\Avecseqbis is unbounded.
			We use the same approach as in the previous item, \ie,~we simulate sub-extraction
			in the logic using the encoding from Figure~\ref{fig:sub-extraction}, thus obtaining $R',I'$.
			Additionally, we ensure that the selected sub-extraction is interval-closed.
			Importantly, taking an interval-closed sub-extraction preserves property \ref{it:star}.
			It remains therefore to write a formula in \msoW which says that
			sets $R',I'$ satisfy
			\vecseqofsets{R'}{I'} is bounded
			and \dimens{\vecseqofsets{R'}{I'}} is unbounded.
			\begin{itemize}[nosep]
				\item The boundedness of \vecseqofsets{R'}{I'} is expressed in \msoW
					by the fact that the complement of $I'$ has to satisfy the negation of \weakU.
					Indeed, this exactly says that the intervals of $\Asetinc'$ are bounded.
				\item Now, since \vecseqofsets{R'}{I'} is bounded then \dimens{\vecseqofsets{R'}{I'}} is unbounded
					\iof the sequence of the number of occurrences of elements of $I'$ between two consecutive elements of $R'$ is unbounded. 
					Consider the set $X$ of all the positions which are either in $R'$ or not adjacent to an element in $I'$
					(which is expressible in \mso).
					Then, since $R'$ and $I'$ satisfy $(\star)$,
					\dimens{\vecseqofsets{R'}{I'}} is unbounded \iof \weakU[X] holds.
			\end{itemize}
	\end{enumerate}
	This completes the reduction of Lemma~\ref{lem:msowdef:dimtoinf} to Lemma~\ref{lem:msowdef:ttoinf_and_dimttoinf}.
\end{proof}


\subsection{Unboundedness of dimensions for sequences tending to infinity}
\label{ssec:level2}
We now prove Lemma~\ref{lem:msowdef:ttoinf_and_dimttoinf}
whence concluding the proof of Theorem~\ref{thm:main}.
The lemma says that
there is a formula \Aform[\Asetreset, \Asetinc] in \msoW
which is true \iof
the vector sequence \vecseqofsets\Asetreset\Asetinc is defined and satisfies
\toinf{\vecseqofsets\Asetreset\Asetinc} and
\dimens{\vecseqofsets\Asetreset\Asetinc} is unbounded.


The main idea of the proof is to relate the unboundedness of the dimension
with a property definable in \msoW, namely the \emph{asymptotic mix} property,
which was already used in \cite{BPT16} to prove undecidability of \msoU.
This is done in Lemma~\ref{lem:asymptotic-mix} below,
for which we need a couple of notions on vector sequences.
\medbreak

Given an infinite set $\Asetindex \subseteq \Nat$ of positions and a number sequence $\Anumseq \in \Nat^\omega$,
the \emph{restriction of \Anumseq to \Asetindex} is the number sequence obtained by
discarding elements of index not in $\Asetindex$ and is denoted \restrict\Anumseq\Asetindex.
Two number sequences \Anumseq and \Anumseqbis are called \emph{asymptotically equivalent},
denoted by $\Anumseq\asymequiv\Anumseqbis$,
if they are bounded on the same sets of positions,
\ie, if for every set \Asetindex,
\restrict\Anumseq\Asetindex is bounded \iof \restrict\Anumseqbis\Asetindex is bounded.
For example, the number sequences
$\Anumseq=1,1,3,1,5,1,7,1,9,\ldots$
and $\Anumseqbis=2,1,4,1,6,1,8,1,10,\dots$
are asymptotically equivalent,
but $\Anumseq$ and $\Anumseqter=1,3,1,5,1,7,1,9,1,\ldots$ are not.

We say that a vector sequence $\Avecseq$ is an \emph{asymptotic mix} of a vector sequence $\Avecseqbis$
if for every $1$-extraction $\Anumseq$ of $\Avecseq$,
there exists a $1$-extraction $\Anumseqbis$ of $\Avecseqbis$
such that $\Anumseq\asymequiv\Anumseqbis$.
In particular, the empty vector cannot occur in \Avecseq or \Avecseqbis.
Note that this is not a symmetric relation.

Given two vector sequences \Avecseq and \Avecseqbis,
we say that \Avecseqbis \emph{dominates} \Avecseq, denoted $\Avecseq\leq\Avecseqbis$,
if for all \ANindex,
the \ANindex-th vectors in both sequences have the same dimension,
and the \ANindex-th vector of \Avecseq is coordinatewise smaller than or equal
to the \ANindex-th vector of \Avecseqbis.
Furthermore, we suppose that all the coordinates in \Avecseq and \Avecseqbis are positive. 

An extraction \Avecseqbis of \Avecseq is said to be \emph{strict},
if at least one coordinate of each vector of \Avecseq has been popped in \Avecseqbis,
\ie, the dimensions of the vectors of \Avecseqbis are pointwise smaller than those of \Avecseq.
In particular, a vector sequence with some vectors of dimension less than $2$ admits no strict extraction,
nevertheless, it may admit strict sub-extractions. 

The following lemma relates the notion of asymptotic mix with the unboundedness of the dimension of a vector sequence,
providing the vector sequence tends towards infinity.
\begin{lemma}
	\label{lem:asymptotic-mix}
	Let \Avecseq be a vector sequence such that \toinf\Avecseq.
	Then, \dimens\Avecseq is unbounded
	\iof
	there exists a sub-sequence \Avecseqter
	and a strict extraction \Avecseqbis of \Avecseqter satisfying:
	\begin{description}
		\item[{\namedlabel%
					[{\textcolor{darkgray}{\textbf{($P_{\Avecseqbis,\Avecseqter}$)}}}]%
					{eq:dominate->asympt-mix}%
					{($P_{\Avecseqbis,\Avecseqter}$)}%
				}%
			]
			for every $\Avecseqter'\leq\Avecseqter$,
			there exists $\Avecseqbis'\leq\Avecseqbis$
			such that
			$\Avecseqter'$ is an asymptotic mix of $\Avecseqbis'$.
	\end{description}
\end{lemma}
\begin{proof}
	The if direction follows from \cite[Lemma~2.2]{BPT16}, while the converse is new to this work.
	More precisely, it is proved in \cite[Lemma~2.2]{BPT16}
	that given two vector sequences \Avecseqter and \Avecseqbis tending towards infinity,
	if \Avecseqter and \Avecseqbis have bounded dimension,
	then the two following properties are equivalent:
	\begin{description}[topsep=1ex,labelsep=.2em]
		\item[{\namedlabel[{\textcolor{darkgray}{\textbf{(1)}}}]{p1}{(1)}}]
			for infinitely many $i$'s,
			the $i$-th vector of \Avecseqter has higher dimension than the $i$-th vector of \Avecseqbis;
		\item[{\namedlabel[{\textcolor{darkgray}{\textbf{(2)}}}]{p2}{(2)}}]
			the negation of \ref{eq:dominate->asympt-mix}:
			there exists $\Avecseqter'\leq\Avecseqter$
			which is not an asymptotic mix of any $\Avecseqbis'\leq\Avecseqbis$.
	\end{description}
	Let \Avecseq be a vector sequence which tends towards infinity
	and suppose that it has bounded dimension.
	Let \Avecseqter be a sub-sequence of \Avecseq
	and \Avecseqbis be a strict extraction of \Avecseqter.
	By definition, \Avecseqter and \Avecseqbis satisfy \ref{p1},
	have bounded dimension
	and tend towards infinity.
	Therefore, by \cite[Lemma~2.2]{BPT16},
	they satisfy \ref{p2},
	\ie,
	the negation of \ref{eq:dominate->asympt-mix}.
	\medbreak

	We now prove the only if direction.
	We suppose that \Avecseq tends towards infinity and that \dimens\Avecseq is unbounded.
	There exists a sub-sequence \Avecseqter of \Avecseq such that
	every vector occurring in \Avecseqter has dimension at least $2$
	and \dimens\Avecseqter tends towards infinity.

	Define \Avecseqbis as being the extraction obtained from \Avecseqter
	by dropping the last coordinate in each vector
	and observe that it is a strict extraction of \Avecseqter.
	We prove now a result stronger than required,
	by exhibiting a universal $\Avecseqbis'\leq\Avecseqbis$
	which makes \ref{eq:dominate->asympt-mix}
	true whatever the choice of $\Avecseqter'\leq\Avecseqter$,
	hence inverting the order of the universal and the existential quantifiers.
	Let $\Avecseqbis'$ be the vector sequence defined as follows:
	for each position \ANindex,
	denoting
	the \ANindex-th vector of $\Avecseqter$
	by $\vect{\Aveccoord_1,\Aveccoord_2,\ldots,\Aveccoord_k}$,
	the \ANindex-th vector of $\Avecseqbis'$ is set to
	$
	\vect{%
		\min\left(\Aveccoord_1,1\right),
		\min\left(\Aveccoord_2,2\right),
		\ldots,
		\min\left(\Aveccoord_{k-1},k-1\right)
	}
	$
	.
	For instance:
	\[
		\begin{array}{rlllllll}
			\Avecseqter=&
			(2,1),&(3,2,1),&(4,3,2,1),&(5,4,3,2,1),&(6,5,4,3,2,1),&\ldots\\
			\Avecseqbis=&
			(2),&(3,2),&(4,3,2),&(5,4,3,2),&(6,5,4,3,2),&\ldots\\
			\Avecseqbis'=&
			(1),&(1,2),&(1,2,2),&(1,2,3,2),&(1,2,3,3,2),&\ldots
		\end{array}
	\]

	We now prove that every $\Avecseqter'\leq\Avecseqter$ is an asymptotic mix of $\Avecseqbis'$,
	that is, for every $1$-extraction $\Anumseqter$ of $\Avecseqter'$, there exists a $1$-extraction $\Anumseqbis$ of $\Avecseqbis'$
	such that $\Anumseqter\asymequiv\Anumseqbis$.
	We fix an arbitrary vector sequence $\Avecseqter'\leq\Avecseqter$
	and a $1$-extraction $\Anumseqter$ of $\Avecseqter'$.
	We define \Anumseqbis as the number sequence
	whose \ANindex-th number,
	for each position \ANindex,
	is the maximal coordinate in the \ANindex-th vector of $\Avecseqbis'$
	which is less than or equal to the \ANindex-th number in \Anumseqter.
	This coordinate always exists since the first coordinate of every vector in the sequence \Avecseqbis' exists and is equal to $1$.
	For instance, according to the previous example and given $\Avecseqter'$ and \Anumseqter,
	the number sequence \Anumseqbis is defined as follows:
	\[
		\begin{array}{rcccccc}
			\Avecseqter'=&
			(1,1),&(2,1,1),&(3,2,1,1),&(4,3,2,1,1),&(5,4,3,2,1,1),&\ldots\\
			\Anumseqter=&
			1,&2,&1,&4,&1,&\ldots\\
			\Avecseqbis'=&
			(1),&(1,2),&(1,2,2),&(1,2,3,2),&(1,2,3,3,2),&\ldots\\
			\hline
			\Anumseqbis=&
			1,&2,&1,&3,&1,&\ldots
		\end{array}
	\]
	Observe that, by definition, for every position \ANindex,
	the \ANindex-th number of \Anumseqbis is less than or equal to
	the \ANindex-th number of \Anumseqter.
	Hence, for every set \Asetindex of positions,
	if \restrict\Anumseqter\Asetindex is bounded then so is \restrict\Anumseqbis\Asetindex.
	Suppose now that for some set \Asetindex of positions,
	\restrict\Anumseqter\Asetindex is unbounded and fix an integer \ANinteger.
	We want to prove that some number in \restrict\Anumseqbis\Asetindex is greater than \ANinteger.
	Since both \Avecseqter and \dimens\Avecseqter tend towards infinity,
	we can find a threshold $\Athreshold\in\Nat$
	such that for every position $\ANindex\geq\Athreshold$,
	the \ANindex-th vector of \Avecseqter has dimension greater than \ANinteger
	and furthermore all its coordinates are greater than \ANinteger.
	Thus, for all positions $\ANindex\geq\Athreshold$,
	the \ANinteger-th coordinate of the \ANindex-th vector of $\Avecseqbis'$ is defined
	and it is equal to \ANinteger.
	Therefore,
	for all $\ANindex\in\Asetindex$ with $\ANindex\geq\Athreshold$
	and such that the \ANindex-th number in \Anumseqter is at least \ANinteger,
	the \ANindex-th number in \Anumseqbis is at least \ANinteger as well.
	Hence, \restrict\Anumseqbis\Asetindex is unbounded.
\end{proof}

To conclude the proof of Lemma~\ref{lem:msowdef:ttoinf_and_dimttoinf},
it remains to express in \msoW the property stated in Lemma~\ref{lem:asymptotic-mix}:
given two sets \Asetreset and \Asetinc
encoding a vector sequence \vecseqofsets\Asetreset\Asetinc
which tends towards infinity,
there exists a sub-sequence \Avecseqter of \vecseqofsets\Asetreset\Asetinc
and a strict extraction \Avecseqbis of \Avecseqter satisfying \ref{eq:dominate->asympt-mix}.
Note that we can check in \msoW
that a sequence~$\vecseqofsets\Asetreset\Asetinc$
(defined by two set variables~$\Asetreset$ and~$\Asetinc$)
tends towards infinity.
Indeed, this holds
\iof,
for any bounded set~$Y$ (\ie, such that~$\neg\weakU(Y)$),
ultimately,
the intervals of~$\Asetinc$
contain at least two elements of~$Y$.

Given a vector sequence \Avecseq and a set $S$,
we say that an encoding of \Avecseq is \emph{$S$-synchronised}
if it is of the form $S,J$ for some set $J$.
Sub-sequences and strict extractions can be simulated in \mso
by defining subsets of the original encoding,
according to the picture in Figure~\ref{fig:sub-extraction}.
In particular,
given \Asetreset and \Asetinc,
one can construct in \mso
any sub-sequence \Avecseqter of \vecseqofsets\Asetreset\Asetinc
and any strict extraction \Avecseqbis of \Avecseqter,
such that the encodings of \Avecseqter and \Avecseqbis are $S$-synchronised for some $S\subseteq R$.

Now, we express \ref{eq:dominate->asympt-mix} on \Avecseqter and \Avecseqbis encoded by $S,J$ and $S,J'$ respectively.
Every dominated sequence $\Avecseqter'$ (\resp $\Avecseqbis'$) of \Avecseqter (\resp \Avecseqbis)
can be obtained by deleting some rightmost consecutive positions of intervals of $J$ (\resp $J'$),
keeping at least one position for each interval.
This can be done in \mso, preserving the $S$-synchronisation, as depicted below:
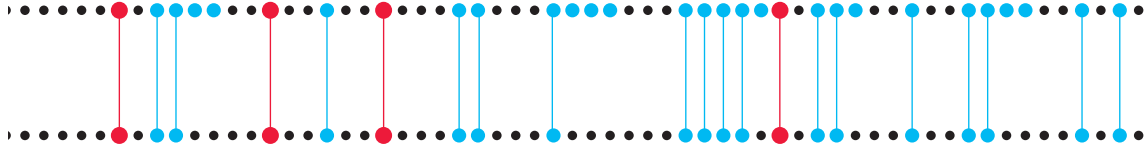
\smallbreak

\noindent\mypic{11}

Finally, we prove that the property saying that $\Avecseqter'$ is an asymptotic mix of $\Avecseqbis'$
can be expressed in \msoW
(using the $S$-synchronisation of the encodings of the two sequences).
Indeed, $1$-extractions can be simulated in \mso
according to the picture in Figure~\ref{fig:sub-extraction},
still preserving the $S$-synchronisation. 
Consider then $K,K'$ such that $S,K$ encodes a $1$-extraction of $\Avecseqter'$ and $S,K'$ encodes a
$1$-extraction of $\Avecseqbis'$,
\ie, such that between every two consecutive positions of $S$,
the elements of $K$ (\resp $K'$) form an interval.
It remains to prove that 
we can express in \msoW that 
\vecseqofsets SK and \vecseqofsets S{K'},
viewed as number sequences, are asymptotically equivalent.
We select a subset of indices of the two number sequences \vecseqofsets SK and \vecseqofsets S{K'}
by selecting elements of $S$
and keeping the maximal intervals in $K$ and $K'$ that directly follow the selected elements.
Then, we check thanks to \weakU
that the corresponding sequences of numbers
are both bounded or both unbounded.
This process is depicted here:
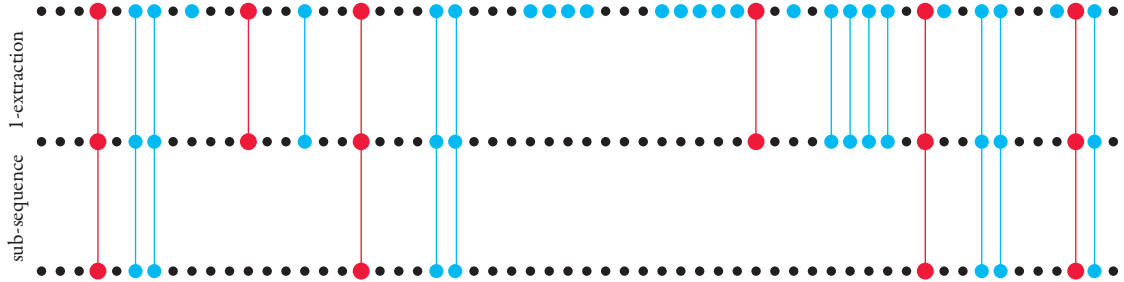
\smallbreak

\noindent\mypic{12}


\section{How MSO+\ensuremath{\texttt U} may talk about ultimate periodicity}
\label{sec:projection}
We consider now the quantifier \quantifierP defined in the introduction.
Recall that a set of positions $X\subseteq\Nat$ is called \emph{ultimately periodic}
if there is some period $\Aperiod\in\Nat$ such that for sufficiently large positions $x\in\Nat$,
either both or none of $x$ and $x+\Aperiod$ belong to $X$.
We consider the logic \msoP,
\ie, \mso augmented with \emph{the quantifier \quantifierP} that ranges over ultimately periodic sets:
\begin{equation*}
	\quantifierP[X]\Aform[X]:
	\text{``the formula \Aform[X] is true for all ultimately periodic sets $X$''}.
\end{equation*}
\begin{example}
	\label{ex:uperiodic}
	The language of $\omega$-words over $\set{a,b}$
	such that the positions labeled by $a$ form an ultimately periodic set
	is definable in \msoP by the following formula:
	\begin{equation*}
		\exists X
		\quad
		\underbrace{\vphantom{\Big(}\forall x \quad \big(x\in X \iff a(x)\big)}_{\text{$X$ is the set of positions labeled by $a$}}
		\quad
		\wedge
		\qquad
		\underbrace{\neg\Big(\quantifierP Y\quad\big(X\neq Y\big)\Big)}_{\text{$X$ is ultimately periodic}}
		\text.
	\end{equation*}
\end{example}

From Example~\ref{ex:uperiodic}, one can see that our quantifier \quantifierP extends strictly the expressivity of \mso, but it is \apriori not clear whether it has decidable or undecidable satisfiability. However, using Theorem~\ref{thm:main}, we can prove Theorem~\ref{thm:periodic} which states that satisfiability over $\omega$-words is undecidable for \msoP.

\begin{proof}[Proof of Theorem~\ref{thm:periodic}]
	Since \msoU has undecidable satisfiability \cite{BPT16},
	the same follows for \msoW by Theorem~\ref{thm:main}.
	Moreover, the predicate \weakU can be defined in terms of ultimate periodicity: 
	a set $X$ of positions satisfies \weakU[X]
	\iof
	for all infinite ultimately periodic sets $Y$,
	there are at least two positions in $Y$
	that are not separated by a position from $X$.
	It follows
	that every sentence of \msoW can be effectively rewritten
	into a sentence of \msoP
	which is true on the same $\omega$-words.
	Hence,
	Theorem~\ref{thm:periodic} follows.
\end{proof}

The predicate \weakU
can thus be expressed in \mso augmented with the quantifier \quantifierP.
It is not clear that the converse is true, and we leave this as an open problem.
However, in Theorem~\ref{thm:projections} below,
we show that,
up to a certain encoding, all languages expressible in \msoP can be expressed in \msoU.

Let $\Sigma$ be a finite alphabet,
\freshletter be a symbol not belonging to $\Sigma$,
and $\Sigma_\freshletter$ denote $\Sigma\cup\set\freshletter$.
We define
$\pi_\Sigma: {\Sigma_\freshletter}^\omega\to\Sigma^\omega$
to be the function which erases all appearances of \freshletter.{}
This is a partial function,
because it is only defined on $\omega$-words over $\Sigma_\freshletter$
that contain infinitely many letters from $\Sigma$.
We extend $\pi_\Sigma$ to languages of $\omega$-words in a natural way.
\begin{theorem}\label{thm:projections}
	Every language definable in \msoP over $\Sigma$
	is equal to $\pi_\Sigma(L)$
	for some language $L\subseteq{\Sigma_\freshletter}^\omega$ definable in \msoU.
\end{theorem}

The main difficulty in expressing ``ultimate periodicity''
is to check the existence of a constant, namely the period,
which is ultimately repeated.
Such a property was already the crucial point in the proof of the undecidability of \msoU \cite{BPT16}.
Indeed the authors managed to express
that an encoded vector sequence (as in Figure~\ref{fig:encoding})
tending towards infinity
has \emph{``ultimately constant dimension''},
\ie, all but finitely many vectors of the sequence have the same dimension
(this dimension will represent the period). This is the content of the following lemma which is a direct consequence of \cite[Lemma~3.1]{BPT16}.
\begin{lemma}
	\label{lem:ultimately constant}
	There exists a formula \Aform[R,I] in \msoU
	with two free set variables
	which holds \iof
	\vecseqofsets RI is defined,
	tends towards infinity,
	and \dimens{\vecseqofsets RI} is ultimately constant.
\end{lemma}

In order to use Lemma \ref{lem:ultimately constant}, we will encode a word 	
$w$ over $\Sigma$ by adding factors of consecutive \freshletter
between every two letters from $\Sigma$,
such that the lengths of those factors tend towards infinity.
%
Let $w=w_1w_2\cdots$ be an $\omega$-word over $\Sigma$
where the $w_i$'s are letters,
and let $\Anumseq=n_1,n_2,\ldots$ be a number sequence.
We define the $\omega$-word
$w_\Anumseq$ over $\Sigma_\freshletter$
as $w_1\freshletter^{n_1}w_2\freshletter^{n_2}\cdots$.
In particular $\pi_\Sigma(w_\Anumseq)=w$.
Now, given a sentence \Aform of \msoP over $\Sigma$,
we consider the language $L$ over $\Sigma_\freshletter$
of all the $\omega$-words of the form $w_\Anumseq$
for $w$ satisfying \Aform
and \Anumseq a number sequence tending towards infinity.
As
we can easily check in \msoU
that the sequence of lengths of the factors of \freshletter
tends towards infinity,
Theorem~\ref{thm:projections} follows immediately from the following lemma.
\begin{lemma}
	For every sentence \Aform in \msoP,
	one can effectively construct a formula $\Aform_\freshletter$ in \msoU
	such that 
	for every $\omega$-word $w\in\Sigma^\omega$,
	the following conditions are equivalent:
	\begin{enumerate}[nosep]
		\item \Aform is true in $w$;
		\item $\Aform_\freshletter$ is true in $w_f$ for every $f$ which tends towards infinity.
	\end{enumerate}
\end{lemma}
\begin{proof}
	Let \Aform be a sentence in \msoP, and $f$ some number sequence tending
	towards infinity.
	Given a word $w$, we construct $w_f$ as explained above.
	
	We construct $\Aform_\freshletter$ by induction on the structure of 
	\Aform, and show that at every step of the induction, \Aform is true in $w$
	if and only of $\Aform_\freshletter$ is true in $w_f$.
	To do so, we use the natural mapping from positions of $w$ into positions of $w_f$
	to handle free variables.
	In particular, every free set variable used in the induction hypothesis
	is supposed to contain only positions not labeled by \freshletter.
	%
	Every formula not containing the quantifier $\quantifierP$ is translated
	straightforwardly, only ignoring all \freshletter positions when counting
	(a bounded number, of course).
	For example, the successor relation will be translated by
	''the first position to the right not labelled with \freshletter''.
	It is routine to check that the satisfiability is the same
	in both $w$ and $w_f$, thanks to the natural mapping described above.
	The key point of the induction
	is the case of a formula $\quantifierP[X]\Aform[X]$.
	We translate it to a formula of the form: $\forall Y, 
	\mathrm{UP}_\freshletter(Y) \Rightarrow
	\Aform_\freshletter(Y)$, where $\mathrm{UP}_\freshletter(Y)$ holds 
	for all sets $Y$ which contain only positions not labeled by $\freshletter$
	and which are ``ultimately periodic when ignoring \freshletter's''.

	Checking if $Y$ does not contain any position labeled by $\freshletter$
	is done by the formula $\forall y, y\in Y \Rightarrow \neg \freshletter(y)$.
	It remains to express the property ``$Y$ is ultimately periodic when ignoring \freshletter's''
	as a formula of \msoU with one free set variable $Y$ which contains only positions not labeled by \freshletter.
	It can be done in the following way:
	``$Y$ is ultimately periodic when ignoring \freshletter's'' \iof
	there exist two sets $R$ and $I$ such that:
	\begin{itemize}
		\item $I$ is exactly the set of all the $\freshletter$,
			$R$ is infinite and disjoint from $I$
			(hence \vecseqofsets RI is defined); 
		\item \dimens{\vecseqofsets RI} is ultimately constant
			(this constant dimension will represent the period of $Y$); 
		\item for every infinite set $R'$ disjoint from $I$,
			alternating with $R$
			(\ie, such that there is exactly one element of $R'$ between two consecutive elements of $R$)
			and such that \dimens{\vecseqofsets{R'}{I}} is ultimately constant, we have that
			all elements of $R'$ are either ultimately in $Y$ or ultimately not in $Y$.
	\end{itemize}
	These properties can be expressed in \msoU
	thanks to Lemma~\ref{lem:ultimately constant},
	since, as $f$ tends towards infinity, the vector sequence \vecseqofsets RI,
	tends towards infinity. It is easy to check that the conjunction
	of these three properties is equivalent to 
	``$Y$ is ultimately periodic when ignoring \freshletter's''.
  %
\end{proof}


\section{Conclusion}
\label{sec:conclusion}
We have considered the extensions of \mso with the following features
(two quantifiers and two second-order predicates):
\begin{description}[align=right,leftmargin=5.5em,labelwidth=3em,labelsep=0em,itemindent=-1em,itemsep=.5ex,topsep=.25ex]
	\item[{$\quantifierU[X]\Aform[X]$}]
		\quad the formula \Aform[X] is true for sets $X$ of arbitrarily large finite size.
	\item[{\weakU[X]}] 
		\quad for all $k\in\mathbb{N}$, there exist two consecutive positions of $X$ at distance at least $k$.
	\item[{\strongU[R,I]}]
		\quad the sequence of numbers encoded by $\Asetreset,\Asetinc$ (Figure~\ref{fig:first-encoding}) is defined and unbounded.
			\item[{$\quantifierP[X]\Aform[X]$}]
		\quad the formula \Aform[X] is true for all ultimately periodic sets $X$.
\end{description}
\smallbreak

We have proved that \msoU, \msoW and \msoS
have the same expressive power,
and the translations between these logics are effective
(Theorem~\ref{thm:main}, Lemmas~\ref{lem:u-predicate} and~\ref{lem:u-predicates}).
Furthermore, all three have undecidable satisfiability on $\omega$-words by \cite{BPT16}.
Moreover we have seen that \weakU is expressible in \msoP,
while the converse is also true up to some encoding (Theorem~\ref{thm:projections}).
As a consequence, \mso extended with the quantifier \quantifierP{}
has undecidable satisfiability (Theorem~\ref{thm:periodic}).

We believe that \msoW can be reduced to many extensions of \mso,
in a simpler way than reducing \msoU. 
It is for instance the case
for \mso augmented with the quantifier that ranges over periodic sets
(and not necessarily ultimately periodic sets).



\bibliographystyle{plain}
\bibliography{bdgps17}
\end{document}